\documentclass[a4paper,preprintnumbers,floatfix,superscriptaddress,aps,10pt,twocolumn,notitlepage,longbibliography,accepted=2025-08-06]{quantumarticle}
\pdfoutput=1
\usepackage[utf8]{inputenc}
\usepackage[T1]{fontenc}
\usepackage{amsmath,amsthm,amsfonts,amssymb,mathtools,dsfont}
\usepackage{marvosym}
\usepackage{graphicx}
\usepackage{xcolor}
\usepackage{appendix}
\usepackage{hyperref}
\usepackage[nolist]{acronym}
\usepackage[capitalise]{cleveref}
\usepackage{algorithm}
\usepackage[noend]{algorithmic}
\usepackage{scalerel}
\usepackage{ifthen}

\newtheorem{theorem}{Theorem}
\newtheorem{corollary}[theorem]{Corollary}
\newtheorem{lemma}[theorem]{Lemma} 

\newtheorem{definition}[theorem]{Definition}
\floatname{algorithm}{Protocol}
\crefname{algorithm}{Protocol}{Protocols}

\newtheorem*{theorem1}{Theorem 4}
\newenvironment{proofsketch}{%
  \proof}{\endproof}

\newcommand{\e}{\ensuremath\mathrm{e}} % Euler e
\renewcommand{\i}{\ensuremath\mathrm{i}} % imaginary unit
\DeclareMathOperator{\Tr}{Tr} % trace 
\DeclareMathOperator{\tr}{Tr} 
\renewcommand{\Re}{\operatorname{Re}} % real part
\renewcommand{\Im}{\operatorname{Im}} % imaginary part
\DeclareMathOperator{\U}{U}%unitary group U(n) 
\renewcommand{\L}{\mathcal{L}} % linear operators

\newcommand{\CC}{\mathbb{C}}% complex numbers
\newcommand{\RR}{\mathbb{R}}% reals

\newcommand{\1}{\mathds{1}} % identity operator
\providecommand{\openone}{\mathds{1}}
\newcommand{\PP}{\mathbb{P}} % Prob.
\renewcommand{\Pr}{\mathbb{P}}
\newcommand{\mc}[1]{\mathcal{#1}}
\renewcommand{\H}{\mc{H}} % Hilbert space 

\DeclarePairedDelimiterX{\abs}[1]{\lvert}{\rvert}{%
  \ifblank{#1}{\,\cdot\,}{#1}
}   % absolute value

\DeclarePairedDelimiterX\norm[1]\lVert\rVert{%
  \ifblank{#1}{\,\cdot\,}{#1}
}   % norm

\newcommand{\vbt}[1]{{\ttfamily #1}} 

\DeclarePairedDelimiterX\Set[1]\{\}{%
  
  #1
}

\DeclarePairedDelimiter{\bra}{\langle}{\vert}
\DeclarePairedDelimiter{\ket}{\vert}{\rangle}
\DeclarePairedDelimiterX\braket[2]{\langle}{\rangle}%
  {#1\kern0.15ex\delimsize\vert\kern0.15ex\mathopen{}#2}
\newcommand{\bk}[2]{\braket{#1}{#2}}
\DeclarePairedDelimiterX\ketbra[2]{\vert}{\vert}%
  {#1\kern0.15ex\delimsize\rangle\delimsize\langle\kern0.15ex\mathopen{}#2}
\newcommand{\kb}[2]{\ketbra{#1}{#2}}

\newcommand{\argdot}{{\,\cdot\,}} % for a dot as an argument
\renewcommand{\vec}[1]{\mathbf{#1}} % for vectors
\DeclareMathOperator{\LandauO}{\mathrm{O}} % Landau big-O notation
\newcommand{\Sgate}{\mathrm{S}}
\newcommand{\Tgate}{\mathrm{T}}
\newcommand{\Hgate}{\mathrm{H}}
\newcommand{\Zgate}{\mathrm{Z}}
\newcommand{\Xgate}{\mathrm{X}}
\newcommand{\iS}{\mathrm{s}}
\newcommand{\iSd}{{\mathrm{s}^{\scaleobj{.8}{{-1}}}}}
\newcommand{\iH}{\mathrm{h}}
\newcommand{\iT}{\mathrm{t}}
\newcommand{\veps}{\varepsilon}
\DeclareMathOperator{\Choi}{J}
\newcommand{\psip}{\psi^\perp}
\newcommand{\phip}{\phi^\perp}
\newcommand{\T}{\intercal} %transpose
\newcommand{\Fid}{\mathrm{F}_{\mathrm{avg}}}
\newcommand{\Pass}{\Pr[\text{\textnormal{\vbt{``pass''}}}]}

\newcommand{\su}{\mathfrak{su}}
\newcommand{\Smodel}{\big(\ketbra{+}{+},\{\Sgate,\Sgate^\dagger\},\{\ketbra{+}{+},\ketbra{-}{-}\}\big)}

\newcommand{\Universalmodel}{\big(\ketbra{+}{+},\{\Sgate,\Sgate^\dagger,\Hgate,\Tgate\},\{\ketbra{+}{+},\ketbra{-}{-}\}\big)}
\usepackage{tikz}
\usetikzlibrary{3d}
\usepackage{tikzpeople}
\usetikzlibrary{arrows,shapes}
 \definecolor{niceblue}{rgb}{0.33,0.5,0.8}
\tikzstyle{blau} = [top color=niceblue!12, bottom color=niceblue!88, 
                    shading = axis, shading angle=-10]%blue shading

\newacro{RB}{randomized benchmarking}
\newacro{GST}{gate set tomography}
\newacro{POVM}{positive operator-valued measure}
\newacro{PVM}{projection-valued measure}
\newacro{CP}{completely positive}
\newacro{CPTP}{completely positive trace preserving}
\newacro{PSD}{positive semidefinite}
\newacro{NISQ}{noisy intermediate-scale quantum}
\newacro{SPAM}{state preparation and measurement} 
\newacro{ONB}{orthonormal basis}

\begin{document}
\title{Classical certification of quantum gates under the dimension assumption}
\author{Jan N\"oller}
\email{jan.noeller@tu-darmstadt.de}
\affiliation{Department of Computer Science, Technical University of Darmstadt, Darmstadt, 64289 Germany}

\author{Nikolai Miklin}
\email{nikolai.miklin@tuhh.de}
\thanks{these authors contributed equally to this work.}
\affiliation{Institute for Quantum Inspired and Quantum Optimization, Hamburg University of Technology, Germany}

\author{Martin Kliesch}
\affiliation{Institute for Quantum Inspired and Quantum Optimization, Hamburg University of Technology, Germany}

\author{Mariami Gachechiladze}
\affiliation{Department of Computer Science, Technical University of Darmstadt, Darmstadt, 64289 Germany}

\begin{abstract} 
The rapid advancement of quantum hardware necessitates the development of reliable methods to certify its correct functioning. 
However, existing certification tests fall short, as they either suffer from systematic errors or do not guarantee that only a correctly functioning quantum device can pass the test.
We introduce a certification method for quantum gates tailored for a practical server-user scenario, where a classical user tests the results of exact quantum computations performed by a quantum server. 
This method is free from the systematic state preparation and measurement (SPAM) errors. 
For single-qubit gates, including those that form a universal set for single-qubit quantum computation, we demonstrate that our approach offers soundness guarantees based solely on the dimension assumption.
Additionally, for a highly-relevant phase gate -- which corresponds experimentally to a $\pi/2$-pulse -- we prove that the method's sample complexity scales as $\mathrm{O}(\varepsilon^{-1})$ relative to the average gate infidelity $\varepsilon$.
By combining the SPAM-error-free and sound notion of certification with practical applicability, our approach paves the way for promising research into efficient and reliable certification methods for full-scale quantum computation.
\end{abstract}
\maketitle

\makeatletter % metadata for hyperref
 \hypersetup{pdftitle = {Classical certification of quantum gates under the dimension assumption},
       pdfauthor = {Jan Nöller, Nikolai Miklin, Martin Kliesch, Mariami Gachechiladze},
       pdfsubject = {Quantum computing},
       pdfkeywords = {% Key words are extremely important for finding this work, e.g., in Google scholar 
              quantum, certification, verification, assumptions, 
              self-testing, Bell, test, device, independent, independence, 
              verifier, user, remote, computer, computation, 
              dynamics, gate, circuit, random, sequence, algorithm, protocol, 
              average, fidelity, diamond, norm, 
              sampling, complexity, 
              gauge, freedom, NISQ
              }
      }
\makeatother

\section{Introduction}
The problem of certifying the correct functioning of quantum devices is crucial for developing quantum hardware and has naturally evolved into a field of study known as quantum system characterization (see Refs.~\cite{Eisert2020QuantumCertificationAnd,Kliesch2020TheoryOfQuantum} for reviews).
Particularly challenging is assessing the quality of quantum gates due to unavoidable \ac{SPAM} errors~\cite{Kliesch2020TheoryOfQuantum}.
They are a limiting factor in standard quantum process tomography \cite{Chuang97PrescriptionForExperimental,Mohseni08QuantumProcessTomography} and direct certification methods \cite{Liu2020EfficientVerificationOf,Zhu2020EfficientVerificationOf,Zeng2020QuantumGateVerification}. 
Two broad families of characterization methods have been developed to address this challenge: \ac{GST} \cite{MerGamSmo13,BluGamNie13,Brieger21CompressiveGateSet} and \ac{RB} \cite{EmeAliZyc05,Levi07EfficientErrorCharacterization,DanCleEme09, EmeSilMou07,KniLeiRei08,Magesan2012,Helsen20AGeneralFramework} with its many variants 
(see Refs.~\cite{Helsen20AGeneralFramework,Heinrich22GeneralGuarantees} for a recent overview).
Both types of protocols can be used to estimate gate errors in a \ac{SPAM}-robust manner by executing sequences of gates with varying lengths.
However, neither of them is suitable for certification, as they cannot definitively rule out the implementation of incorrect gates.
\ac{GST} provides a set of compatible descriptions of the underlying experiment, some of which may not be connected to the implemented operations by any physical gauge. 
\ac{RB}, in turn, already requires that the implemented gates are close to the target ones in order to interpret the output decay parameters~\cite{Helsen20AGeneralFramework}. 

In general, any certification method should satisfy two key properties: \emph{completeness} and \emph{soundness}~\cite{Kliesch2020TheoryOfQuantum}. 
In simple terms, the former means that a certification test should accept correctly implemented target quantum operations, and the latter that \emph{only} the correct implementation should be accepted.
While existing \ac{SPAM}-robust characterization methods satisfy completeness, currently, there is no method in the quantum characterization literature that is both sound and is free from \ac{SPAM} errors. 

An independent line of research, known as \emph{self-testing}~\cite{mayers2003self}, offers a framework for sound certification of quantum devices while treating them as black boxes. 
The black-box nature of certification ensures that all quantum operations, including the state preparation, quantum evolution, and the measurement are certified at the same time. 
In the context of quantum gate certification, this implies that conclusions that one draws from self-testing are free from \ac{SPAM} errors, because any such systematic errors are detected.

The self-testing literature primarily focuses on certifying entangled states and local measurements in the Bell test (see e.g., Ref.~\cite{Supic2019SelfTesting} for a review). 
However, some works have also considered quantum channels~\cite{sekatski2018certifying} or instruments~\cite{Wagner2020Device}, that can be applied to one of the subsystems in the Bell test, and also entangling interactions~\cite{Sarkar2024Model}.
Extending the framework of self-testing to include quantum dynamics, and in particular quantum gates~\cite{magniez2006self,reichardt2013Classical}, makes it relevant for the problem of testing quantum computers. 
The general idea is to combine self-testing of states and measurements with protocols such as process tomography, which typically require trusted devices, to obtain device-independent certification of quantum operations.
However, all these schemes demand an experimental setup consisting of two well-isolated parts, a requirement that is challenging to achieve within a single quantum processor. 
A recent self-testing study~\cite{metger2021self} suggested relaxing this experimental requirement by introducing computational assumptions~\cite{mahadev2018classical}. 
For these assumptions to be accepted, however, experimental capabilities beyond the reach of current quantum hardware are required~\cite{stricker2022towards}.

Finally, a related question of learnability of quantum operations was recently raised in Ref.~\cite{Huang22FoundationsForLearning}. 
There, the authors proposed a framework to investigate the learnability of the intrinsic descriptions of quantum experiments from observational data.
However, the general result obtained there relies on more assumptions than desirable for a certification task.

\begin{figure}[t!]
    \centering
    \includegraphics[scale=.9]{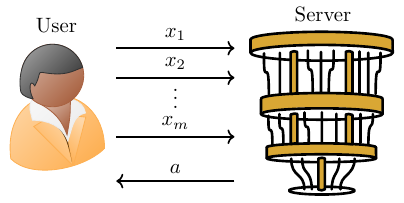}
    \caption{Schematic figure of a server-user interaction. 
    A classical user via a classical channel transmits instructions $x_1$, $x_2$, $\dots$ $x_m$ to a quantum server, which performs a quantum computation and returns the result $a$.}
\label{fig:user-server}
\end{figure}

In this work, our objective is to achieve \ac{SPAM} error-free certification of a quantum computer's correct functioning,  without requiring physical access to it, assuming as little about the quantum computer's internal functioning as possible.
This investigation is particularly relevant in a practical server-user configuration (see \cref{fig:user-server}): a classical user via a classical channel transmits gate sequences to a quantum server, which then implements them and returns the measurement outcomes. 
Since achieving this goal without any assumptions is impossible, we adopt the commonly considered assumption on the quantum system's dimension~\cite{mohan2019sequential,miklin2020semi,tavakoli2020self,Miklin21UniversalScheme,anwer2020experimental,navascues2023self}.

Since we assume that no part of the quantum apparatus, such as the measurement device, is characterized prior to the test, certification is only possible up to the degrees of freedom inherent to quantum mechanics, i.e., unitary or anti-unitary transformation (a unitary and the complex conjugation)~\cite{wigner1931gruppentheorie}.
It is important to note that this is the absolute minimum degree of freedom that cannot be excluded in black-box tests, as it corresponds to the simultaneous change of bases. 
Our method is based on a very intuitive idea of testing outcomes of exact quantum computation for quantum gate sequences that can be resolved efficiently classically.
Examples include gate sequences that compose into the identity gate or simply a zero-length sequence, for which the system is measured directly after the state preparation.
For certification of quantum gates, only the input-output correlations are used and no entanglement with an auxiliary system is required.

Here, we prove that a gate set universal for single qubit quantum operations can be certified within our framework, and analyze in detail the certification of a relevant single-qubit gate, which corresponds to the application a $\pi/2$-pulse.
The latter is routinely implemented on a quantum computer for creating an equal superposition of qubit basis states~\cite{leibfried2003quantum}.
We show that the required \emph{sample complexity}, measured in terms of the number of individual runs of the experiment~\cite{Kliesch2020TheoryOfQuantum}, scales as $\LandauO(\veps^{-1})$ with respect to the average gate infidelity $\veps$. 

This favorable scaling paves the way for a promising research in the certification of quantum systems. 
In this work, we present proof-of-principle results for a number of important single-qubit gates. 
At the same time, the proposed protocol has the potential to be extended for application in full-scale quantum computation.
While these generalizations are intriguing, they may result in less optimal guarantees. 
As a result, the scaling demonstrated in our proofs not only minimizes the resources required to certify single-qubit gates but also sets a benchmark for assessing the effectiveness of future generalizations.

The rest of the paper is organized as follows. 
In \cref{sec:protocol}, we explain the experimental setup, outline the assumptions, and describe the protocol. 
In \cref{sec:results}, we present our results on certification of single-qubit quantum operations, with the main contributions stated in \cref{th:selftest_S}, \cref{cor:sample_S}, and \cref{th:universal}.
Technical details supporting the main claims of the paper are left to the appendix. 

\section{Setup and protocol}
\label{sec:protocol}
\begin{figure}
    \centering
    \includegraphics[scale=1]{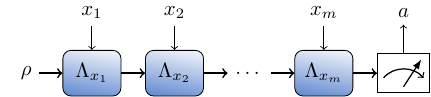}
    \caption{Scenario of the certification test. 
    A system is prepared in state $\rho$ and undergoes a series of transformations, $\Lambda_{x_1}$, $\Lambda_{x_2}$, \dots, $\Lambda_{x_m}$, specified by classical instructions $\Set{x_i}_{i=1}^{m}$ from a fixed set, after which it is measured in a fixed basis, producing outcome $a$.
    The sequences $\vec{x} = x_1x_2\dots{}x_m$ of varied length $m$ are chosen at random from the predetermined set $\mathcal{X}$.
    }
    \label{fig:scenario}
\end{figure}
The experimental setup that we consider is common to many established certification methods, such as \ac{RB}~\cite{Helsen20AGeneralFramework} and \ac{GST}~\cite{MerGamSmo13}.
The setup, or scenario as it would be called in the self-testing literature, is shown schematically in \cref{fig:scenario}.
A quantum system is prepared in some initial state, after which a sequence of quantum gates is applied to it, and it is finally measured in some fixed basis. 
For each gate in a given sequence, a label, chosen from some finite set $X$, is communicated to the quantum computer.
The certification protocol relies on a particular finite subset of sequences $\mathcal{X}\subset X^\ast$, which are determined before the protocol begins. 
In a single repetition of the protocol, a random string $\vec{x}\coloneqq x_1x_2\dots{}x_m$ is selected from $\mathcal{X}$, and after the quantum computer implements the corresponding computation, the outcome $a\in A$ of the measurement is read out, where $A$ is the set of all possible outcomes.
The certification protocol decides to proceed or abort, depending on the deterministic outcome $a_\vec{x}$, corresponding to the ideal implementation of the target quantum computation.
The length $m$ of different sequences $\vec{x}\in\mathcal{X}$ can be different, and, in particular, be zero, which we denote by the empty string $\epsilon$, and by which we mean that the system is measured directly after the state preparation.

Next, we give a definition of a quantum model.
\begin{definition}\label{def:q_model}
    A $d$-dimensional quantum model is a $3$-tuple $(\rho,\{\Lambda_x\}_{x\in X},\{M_a\}_{a\in A})$, consisting of a quantum state $\rho$, prepared at the beginning of the quantum computation, a finite set of quantum channels (gates) $\{\Lambda_x\}_{x\in X}$, from which a quantum circuit is composed, and a \ac{POVM} $\{M_a\}_{a\in A}$, measured at the end of the computation, all defined over a Hilbert space $\H$ with $\dim(\H)=d$.
\end{definition}
For a quantum model involving only unitary quantum channels, we use the corresponding unitary operators in the definition of the model. 
We treat our setup as a single black box, making only the \textit{minimal} assumptions -- such that removing any of these assumptions would render the results of this paper unattainable in the given scenario. 
\begin{itemize}
    \item[(i)] \emph{The dimension assumption.} 
    We assume that in each experimental round, the implemented state preparation is mathematically described by a density operator, the operations are described by \ac{CPTP} maps, and the measurement is described by a \ac{POVM}, all defined over a Hilbert space of a specified dimension. 
    \item[(ii)] \emph{Context independence.} 
    We assume that in each repetition of our protocol, for each label $x$, there is a single corresponding quantum channel $\Lambda_x$ that the quantum computer implements. 
    That is, for any input sequence $\mathbf{x}=x_1x_2\dots x_m$, we assume that the input state undergoes the corresponding series of transformations $\Lambda_{x_m}\circ\dots\circ\Lambda_{x_2}\circ\Lambda_{x_1}$. 
    In the certification literature, this assumption would be referred to as existence of a single-shot implementation function.
\end{itemize}
The above assumptions allow us to mathematically describe an experiment in the considered scenario by a quantum model, as in \cref{def:q_model}.
Other minor assumptions include error-free functioning of the classical part of the quantum computer, such as control circuits, and our ability to randomly select the gate sequences.
Notably, it is not necessary to assume the ability to sample the gates precisely according to a fixed distribution.

The dimension assumption is critical, because a single quantum system, with which a classical user is interacting, can be simulated by a classical system of a larger dimension~\cite{Miklin21UniversalScheme,miklin2020semi}.
We note here, that the dimension assumption also means that no side channels, e.g., operated by an adversary, are considered.
The assumption of context independence cannot be removed, because for each gate sequence one can always assign a \ac{POVM} that reproduces the required statistics, and a quantum computer can simply perform this measurement on the preparation state, ignoring the structure of the gate sequence.

Finally, when we estimate the sample complexity of our protocol, we need to use an assumption (iii) of \emph{independence of repetitions}.
As we argue below, (iii) is a part of (i), but we separate it here for clarity. 
We need the independence of repetitions to treat events in different repetitions of the protocol as statistically independent. 
However, the dimension assumption already implies that there is a tensor product between quantum models implemented in different rounds of the protocol, which leads to the independence of the observed outcomes.
Indeed, if one does not assume this tensor structure, there is always a possibility of a global unitary gauge applied to several copies of state, measurements, and quantum gates in a way that makes these objects entangled across the protocol repetitions.
This would not change the observed statistics, but would make it impossible to say anything about the models implemented in a single experimental run.
Note, however, that for the certification, we \emph{do not} need the assumption that quantum models in different repetitions of the protocol are identical.

We are ready to present our \cref{protocol} for classical certification of quantum gates.
We give a general formulation for a given set of gates $X$, with an important property that among all possible sequences $X^\ast$, there are such $\vec{x}$, for which we can predict the deterministic outcome $a_\vec{x}$, which a noiseless quantum computer should output.  

The basic idea is that if \cref{protocol} accepts a model for some large $N$, and therefore obtained the correct outcome $a_\vec{x}$ for different sequences $\vec{x}$ in all these tests, we can obtain a certain level of confidence, typically denoted by $1-\delta$, that our quantum computer implements the quantum model correctly.  
We define more precisely below what we mean by the latter, building on similar definitions in the self-testing literature~\cite{Miklin21UniversalScheme}.

\begin{algorithm}[H]
\caption{Classical certification of quantum gates}\label{protocol}
\begin{algorithmic}
\STATE Set $N$ -- the number of repetitions, $\mathcal{X}$ -- the set of gate sequences $\vec{x}$, each with the corresponding deterministic outcome $a_\vec{x}$, and $\mu$ -- a probability mass function over $\mathcal{X}$.
\FOR{$i\in[N]$}
\STATE sample $\vec{x}$ from $\mathcal{X}$, according to $\mu$;
\STATE run the quantum circuit consisting of the state preparation, the sequence $\vec{x}$ of gates, and the measurement, record the outcome $a$;
\IF{$a\neq a_{\vec{x}}$}
\STATE{output \vbt{``reject''} and end the protocol;}
\ENDIF
\ENDFOR
\STATE output \vbt{``accept''}.
\end{algorithmic}
\end{algorithm}

Since we do not assume any part of the quantum computer to be characterized, and rely only on the classical data in our certification, any two quantum models which are equivalent up to the symmetries in quantum mechanics~\cite{wigner1931gruppentheorie} will produce the same statistics, and we will not be able to distinguish between them. 
At the same time, we would like to exclude any other quantum model, which is formalized by the following definition.

\begin{definition}\label{def:correct_model}
    For a target $d$-dimensional quantum model $(\rho,\{U_x\}_{x\in X},\{M_a\}_{a\in A})$ with unitary channels, we say that another $d$-dimensional quantum model $(\tilde{\rho},\{\tilde{\Lambda}_x\}_{x\in X},\{\tilde{M}_a\}_{a\in A})$ is its correct implementation if there exists a unitary operator $U$ (with a possible complex conjugation $^{(\ast)}$), such that
    \begin{align}\label{eq:def_correct_model}
    \begin{split}
        \tilde{\rho} &= U\rho^{(\ast)} U^\dagger,\\
        \tilde{\Lambda}_x(\argdot) &= U U_x^{(\ast)} U^\dagger (\argdot) U {U_x^\dagger}^{(\ast)} U^\dagger,\quad \forall x\in X,\\
        \tilde{M}_a &= U M_a^{(\ast)} U^\dagger,\quad \forall a\in A.
    \end{split}
    \end{align}
\end{definition}
The negation of \cref{def:correct_model} provides a definition of an incorrect implementation.
Following the terminology of the self-testing literature~\cite{Supic2019SelfTesting}, we say that the target outcomes $a_{\vec{x}}$ \emph{self-test} a quantum model, if from the fact that the observed outcomes correspond to $a_{\vec{x}}$, we can infer that the implemented quantum model is a correct implementation of the target model, in the sense of \cref{def:correct_model}.
Moreover, we say that the self-test is \emph{robust}, if for small deviations in the outcomes, the target and the implemented models are close in some distance.
For quantum states and quantum measurements, we use the infidelity and the spectral distance, respectively, and for quantum gates, we use the average gate infidelity.
Here, we use the following expression for the average gate fidelity between a qubit channel $\tilde\Lambda$ and qubit unitary channel $\Lambda$ (see e.g.,~\cite{Kliesch2020TheoryOfQuantum}),
\begin{equation}
    \Fid(\tilde\Lambda,\Lambda) = \frac{2}{3}\Tr[\Choi(\tilde\Lambda)\Choi(\Lambda)]+\frac{1}{3},
\end{equation}
where $\Choi(\argdot)$ is the Choi-Jamio\l{}kowski~\cite{Cho75,Jam72} map, defined for a qubit channel $\Lambda$ as
\begin{equation}
\Choi(\Lambda) \coloneqq \frac{1}{2}\sum_{i,j\in\Set{0,1}}\Lambda(\kb{i}{j})\otimes\kb{i}{j}.
\end{equation}

\Cref{def:correct_model} is motivated by the natural symmetries in quantum mechanics~\cite{wigner1931gruppentheorie}: two quantum models connected by a (anti-) unitary transformation as in \cref{eq:def_correct_model} will always lead to the same observed statistics.
For this reason, whenever a collection of deterministic outcomes produced by a quantum computer self-test a target quantum model, we also say that the quantum computer implemented this model \emph{correctly}, even if the target and the implemented models are not exactly equal, but only equivalent up to a (anti-) unitary gauge.
For the same reason, we do not treat the (anti-) unitary gauge in the implemented model as \emph{noise}.

Following the terminology of the certification literature~\cite{Kliesch2020TheoryOfQuantum}, we say that \cref{protocol} is a certification test for a target quantum model with respect to appropriately chosen distances, if the protocol is \emph{complete} and \emph{sound}. 
\begin{definition}
    Given a null hypothesis $H_0$ and an alternative hypothesis $H_1$, a test is complete, if
    \begin{equation}\label{eq:def:complete}
        \PP[\text{\textnormal{\vbt{``accept''}}}\vert H_1] \geq 1-\delta',
    \end{equation}
    for $\delta'<\frac{1}{2}$ and sound, if 
    \begin{equation}\label{eq:def:sound}
        \PP[\text{\textnormal{\vbt{``reject''}}}\vert H_0] \geq 1-\delta,
    \end{equation}
    for $\delta<\frac{1}{2}$.
\end{definition}
It is common to take $\delta'=0$ in \cref{eq:def:complete}, which is also what we do in this work.
In our certification test, $H_1$ is the hypothesis that the implemented model is a correct implementation of a target model, as given by \cref{def:correct_model}.
For our certification results for a gate set universal for single-qubit quantum computation, we take $H_0$ to be the negation of $H_1$.
One can refer to this case as ``non-robust soundness'', as this would require the test to run infinitely. 
In the self-testing literature, this is often referred to as the \emph{ideal} case~\cite{Supic2019SelfTesting}.
For our results for a phase gate, we relax the hypothesis $H_0$, and only exclude models for which no unitary can bring them $\veps$-close in the chosen distance to the target model.  
This allows us to set an upper-bound on the required number of repetitions of \cref{protocol}, i.e., the sample complexity.

\section{Certification of single-qubit quantum models}
\label{sec:results}
In this section, first we prove that \cref{protocol} is an $\veps$-certification test for a quantum model $\Smodel$, where $\Sgate \coloneqq \ketbra{0}{0}+\i\ketbra{1}{1}$.
The soundness of this test follows from a more general robust self-testing-type result, which we prove first.
Here, we use $\Sgate$ gate to model a $\pi/2$-pulse with respect to an \ac{ONB} $\{\ket{+},\ket{-}\}$~\cite{leibfried2003quantum}.
However, all the following results for the quantum model $\Smodel $ also apply for other unitarily equivalent models such as e.g., $(\ketbra{0}{0},\{\sqrt{\Xgate},\sqrt{\Xgate}^\dagger\},\{\kb{0}{0},\kb{1}{1}\})$, where $\sqrt{\Xgate}\coloneqq \ketbra{+}{+}+\i\ketbra{-}{-}$.
Then, we prove an ideal self-testing result for a quantum model with an additional $\Hgate$ gate and the square root of the $\Sgate$ gate, i.e., the model $\Universalmodel$, which shows that a single-qubit universal gate set can be certified using \cref{protocol}. 

\subsection{Certification of a phase gate}
We explain in detail certification of the $\Sgate$ gate, or more precisely, the quantum model
\begin{equation}\label{eq:S_gate_target_model}
  \Smodel\, .
\end{equation}
We set $X=\Set{\iS,\iSd}$ for the classical instructions given to a quantum computer which specify whether it should implement the $\Sgate$ gate or its inverse, respectively.  

Surprisingly, it is sufficient to consider the following set of strings in \cref{protocol}
\begin{equation}\label{eq:X_S}
\mathcal{X} = \Set{\epsilon,\iS\iS,\iS\iSd,\iSd\iS,\iSd\iSd},
\end{equation}
where $\epsilon$ denotes the empty string. 
Next, we set $A = \Set{+,-}$, and for the sequences in~\cref{eq:X_S}, determine that the deterministic outcomes corresponding to the target model are the following
\begin{equation}\label{eq:outcomes_S}
a_\vec{x} = \left\{ 
  \begin{array}{ c l }
    + & \quad \textrm{for } \vec{x}\in \Set{\epsilon, \iS\iSd,\iSd\iS},\\
    - & \quad \textrm{for } \vec{x}\in \Set{\iS\iS,\iSd\iSd}.
  \end{array}
\right.
\end{equation}
It can be easily seen that in case of noiseless implementation of the state preparation $\kb{+}{+}$, the gates $\Sgate,\Sgate^\dagger$, and the measurement $\{\kb{+}{+},\kb{-}{-}\}$, \cref{protocol} always accepts.
In other words, \cref{protocol} is \emph{complete} for certification of the quantum model $\Smodel$, with $\mathcal{X}$ and $\Set{a_\vec{x}}_{\vec{x}\in\mathcal{X}}$ specified by \cref{eq:X_S} and \cref{eq:outcomes_S}, respectively. 

Proving the \emph{soundness} of the protocol is much less straightforward.
To achieve this, we first state and prove the following self-testing result.

\begin{theorem}\label{th:selftest_S}
Let $(\tilde{\rho},\{\tilde{\Lambda}_\iS,\tilde{\Lambda}_\iSd\},\{\tilde{M}_+,\tilde{M}_-\})$ be a single-qubit quantum model that passes with probability at least $1-\veps$ a single repetition of \cref{protocol} with uniform sampling from the gate sequences \eqref{eq:X_S} and for deterministic measurement outcomes \eqref{eq:outcomes_S}. 
Then the quantum model is $\LandauO(\veps)$-close to the target model \eqref{eq:S_gate_target_model}, i.e.,
there is a unitary $U\in\U(2)$ such that
\begin{align}\label{eq:th_selftest_S}
\begin{split}
\Fid(\tilde{\Lambda}_\iS, U\Sgate U^\dagger) &\geq 1-\LandauO(\veps),\\
\Fid(\tilde{\Lambda}_\iSd, U\Sgate^\dagger U^\dagger) &\geq 1-\LandauO(\veps)\,,\\
\Tr[\tilde{\rho}U\kb{+}{+}U^\dagger] &\geq 1-\frac{15}{2}\veps,\\
\norm*{\tilde{M}_{+}-U\kb{+}{+}U^\dagger}_\infty &\leq \frac{5}{2}\veps\, .
\end{split}
\end{align}
\end{theorem}

In the case of unitary channels, we use the respective operators as the argument of the fidelity function for simplicity of notation.
Below, we give a sketch of the proof, and the full proof can be found in \cref{app:th_selftest_S}.
\begin{proofsketch}
The conclusions of the theorem follow from the condition $\Pass\geq 1-\veps$ and the dimension assumption. 
As a first step, we show that for small $\veps$, the measurement effects $\tilde{M}_+$ and $\tilde{M}_-$ are close to being rank-1 projectors, which we denote as $\kb{\psi}{\psi}$ and $\kb{\psip}{\psip}$.
Next, we show that the \acp{POVM} which one obtains by applying the adjoint maps $\tilde{\Lambda}_\iS^\dagger$ and $\tilde{\Lambda}_\iSd^\dagger$ to $\kb{\psi}{\psi}$ and $\kb{\psip}{\psip}$ are also close to be projective for small $\veps$.
We denote the corresponding projectors by $\kb{\phi}{\phi}$ and $\kb{\phip}{\phip}$.
Importantly, we find that $\tilde{\Lambda}_\iS^\dagger(\kb{\psi}{\psi})\approx\kb{\phi}{\phi}$ and $\tilde{\Lambda}_\iSd^\dagger(\kb{\psi}{\psi})\approx\kb{\phip}{\phip}$, and since the adjoint maps of channels are unital, also $\tilde{\Lambda}_\iS^\dagger(\kb{\psip}{\psip})\approx\kb{\phip}{\phip}$ and $\tilde{\Lambda}_\iSd^\dagger(\kb{\psip}{\psip})\approx\kb{\phi}{\phi}$.  
Next, we obtain a partial characterization of the Choi-Jamio\l{}kowski state of the channel $\tilde{\Lambda}_\iS$ in the bases of $\ket{\psi},\ket{\psip}$ and $\ket{\phi},\ket{\phip}$, with the leading terms corresponding to the subspace spanned by $\ket\psi\ket\phi$ and $\ket\psip\ket\phip$.
Then, we show that $\tilde{\Lambda}_\iS$ is close to being a unitary channel, for which we use the conditions $\tilde{\rho}\approx \kb{\psi}{\psi}$ and $\tilde{\Lambda}_\iS(\tilde{\rho})\approx\kb{\phip}{\phip}$. 
Finally, by fixing the global phases of the vectors $\ket{\psi},\ket{\psi^\perp}$ appropriately, we construct a suitable gauge unitary,
\begin{equation}\label{eq:proof_S_U}
U = \kb{\psi}{+}-\i\kb{\psip}{-},
\end{equation}
for which the condition for $\tilde{\Lambda}_\iS$ in \cref{eq:th_selftest_S} is satisfied.
Because we obtain characterization of $\tilde{\Lambda}_\iS$ and $\tilde{\Lambda}_\iSd$ in the same basis, the proof also easily extends to the channel $\tilde{\Lambda}_\iSd$.
The bounds for the state and the measurement in \cref{eq:th_selftest_S} with the chosen unitary are also immediate.
\end{proofsketch}

Interestingly, in the ideal case of $\veps=0$, the self-testing argument can also be made for the set of gate sequences without $\iSd\iSd$, or without $\iS\iS$.
Moreover, the effect that the sampling distribution $\mu$ over $\mathcal X$ plays in \cref{th:selftest_S} is purely in determining the constants in front of $\veps$, with the only requirement that each sequence in $\mathcal{X}$ is chosen with some nonzero probability.

We can use the known relation that connects the average gate fidelity and the diamond distance for an arbitrary qubit channel $\tilde\Lambda$ and a unitary channel $\Lambda$~\cite{Kliesch2020TheoryOfQuantum},
\begin{equation}\label{eq:fidelity_and_diamond}
\norm*{\tilde\Lambda-\Lambda}_\diamond \leq 2\sqrt{6}\sqrt{1-\Fid(\tilde\Lambda,\Lambda)},
\end{equation}
to reformulate \cref{th:selftest_S} for the diamond distance with an upper bound of $\LandauO(\sqrt{\veps})$.

Next, using \cref{th:selftest_S}, we show that \cref{protocol} is sound for certification of the model $\Smodel$, as stated by the following corollary.
This corollary is the main practical result of the paper.

\begin{corollary}
\label{cor:sample_S}
\cref{protocol} with uniform sampling from the gate sequences \eqref{eq:X_S} and for deterministic measurement outcomes \eqref{eq:outcomes_S} is an $\veps$-certification test for the $\Sgate$ gate and its inverse with respect to the average gate infidelity, as well as initial state $\kb{+}{+}$ and measurement $\{\kb{+}{+},\kb{-}{-}\}$ with respect to infidelity and spectral norm from $N$ independent samples for $N\geq N_0$ with
\begin{equation}\label{eq:sample_coml_S}
N_0 = \LandauO(\veps^{-1})\ln(\delta^{-1})
\end{equation} 
with confidence at least $1-\delta$.
Moreover, \cref{protocol} accepts the target model $\Smodel$ with probability $1$.
\end{corollary}  

The lower bound of $N_0$ in \cref{cor:sample_S} should be interpreted as a sufficient number of repetitions of \cref{protocol} to reach the target confidence level $1-\delta$.

\begin{proof}
We can invert the statement of \cref{th:selftest_S}, and obtain that for a noisy quantum model, for which there does not exist a unitary $U\in\U(2)$, satisfying \cref{eq:th_selftest_S}, the probability of passing a single repetition of \cref{protocol} is $\Pass\leq 1-\veps$.
If we take $N$ independent copies of such noisy $\tilde{\Lambda}_\iS$, $\tilde{\Lambda}_\iSd$, as well as $\tilde{\rho}$ and $\{\tilde{M}_+,\tilde{M}_-\}$, the probability of \cref{protocol} accepting them is 
\begin{equation}\label{eq:prob_accept_S}
\Pr[\text{\textnormal{\vbt{``accept''}}}] \leq (1-\veps)^N.
\end{equation} 
For the target confidence level $1-\delta$, the acceptance probability in \cref{eq:prob_accept_S} should be upper-bounded by $\delta$, which leads to a lower bound on $N$,
\begin{equation}
N\geq \frac{\ln(\delta^{-1})}{\ln\frac{1}{1-\veps}}.
\end{equation} 
Approximating the logarithm function for $\frac{1}{1-\veps}$, and rescaling $\veps$ such that the lower bounds on the average gate fidelity in \cref{eq:th_selftest_S} are exactly $1-\veps$, leads to the sample complexity stated in \cref{eq:sample_coml_S}. 
\end{proof}

As \cref{cor:sample_S} demonstrates, our method for certification of quantum gates is as efficient as the direct certification of quantum processes~\cite{Liu2020EfficientVerificationOf}, which requires trust in state preparations and measurements (and hence is not free from \ac{SPAM} errors) and also requires an auxiliary system to prepare the Choi-Jamio\l{}kowski state of the process.
The only price to pay is a possibly larger constant factor, which we numerically estimate for the uniform $\mu$ in the next section.

The commonly used \ac{SPAM}-robust characterization methods, \ac{RB} and \ac{GST}, are sample efficient, however, as mentioned earlier, if used for certification, they do not come with the soundness guarantees. 
The proposal of Ref.~\cite{magniez2006self} for self-testing of quantum gates in the Bell test is sound, but
can tolerate very little amount of noise (see discussion in Ref.~\cite{sekatski2018certifying}).
The proposal of Ref.~\cite{sekatski2018certifying} reports higher noise tolerance, but as noted in Ref.~\cite{govcanin2022sample}, self-tests based on violations of a Bell inequality that do not reach the algebraic maximum suffer from a quadratically worse scaling of sample complexity with respect to $\veps$.

Finally, since \cref{protocol} uses the same experimental setup as \ac{GST} and \ac{RB}, it can be seamlessly integrated into these protocols.
In particular, if $\Sgate$ and $\Sgate^\dagger$ are included in the gate set of a \ac{GST} experiment, the statistics gathered there can be used to obtain a lower bound on the average gate fidelity from \cref{th:selftest_S} (under an additional i.i.d.~assumption).  
Moreover, the gauge freedom in the \ac{GST} output can be reduced to unitary for the gates $\Sgate$ and $\Sgate^\dagger$. 
Similarly, if the acceptance probability of \cref{protocol} is estimated before conducting an \ac{RB} experiment, it can be used in the guarantees of the \ac{RB} protocol's output.

\subsection{Numerical investigations}

In this subsection, we supplement our theoretical results of \cref{th:selftest_S} and \cref{cor:sample_S} by numerical investigations.
This also allows us to estimate the coefficients of the linear scaling in \cref{th:selftest_S}, which then translates to an explicit formula for the sample complexity in \cref{cor:sample_S}.
The results of our numerical investigations are shown in \cref{fig:numerics}.

For each randomly generated quantum model $(\tilde{\rho},\{\tilde{\Lambda}_\iS,\tilde{\Lambda}_\iSd\},\{\tilde{M}_+,\tilde{M}_-\})$, we calculate the probability of it failing a single repetition of the protocol, which corresponds to $\veps$ in the statement of \cref{th:selftest_S}.
We then apply a unitary to the target model, which explicitly depends on the noisy random model, as described in the proof of \cref{th:selftest_S}, and calculate the distance between the noisy model and the target model. 
As the latter, we take the maximum of the two average gate infidelities for the $\Sgate$ and $\Sgate^\dagger$ gates.

We consider four different noise models: unitary noise, two combinations of unitary noise with depolarizing noise, and the depolarizing noise.
To generate random noisy quantum models, we apply independent unitaries to the target state, channels, and the measurement, i.e., we take $\tilde{\rho} = U_1\kb{+}{+}U_1^\dagger$, $\tilde{\Lambda}_\iS(\argdot) = U_2\Sgate U_3(\argdot)U^\dagger_3\Sgate^\dagger U^\dagger_2$, $\tilde{\Lambda}_\iSd(\argdot) = U_4\Sgate^\dagger U_5(\argdot)U_5^\dagger\Sgate U^\dagger_4$, $\tilde{M}_+ = U_6\kb{+}{+}U^\dagger_6$, and $\tilde{M}_- = \1-\tilde{M}_+$, for randomly sampled $U_1 \dots U_6$.
Sampling Haar-random unitaries $U_1 \dots U_6$ would result in a quantum model far from the target one, which would not be useful for our numerical investigation.
Therefore, we generate each $U_i$, $i\in \{1,\dots,6\}$, by randomly sampling $u_i\in \su(2)$ with $\norm{u_i}_\infty=1$ and then setting $U_i = \e^{\alpha_i u_i}$ for uniformly sampled $\alpha_i\in[0,1]$.

In addition to the unitary noise, we consider adding various amounts of depolarizing noise to the channels $\tilde{\Lambda}_\iS$ and $\tilde{\Lambda}_\iSd$.
Finally, we consider a noise model in which the state preparation, measurement, and the $\Sgate$-gate are implemented ideally, while the $\Sgate^\dagger$-gate suffers from the depolarizing noise, $\tilde{\Lambda}_\iSd(\rho) = (1-\eta)\Sgate^\dagger(\rho)\Sgate + \eta\Tr[\rho]\frac{\1}{2}$, with the depolarizing parameter $\eta\in [0,1]$. 
A straightforward calculation shows that for this model $1-\Pass = \frac{(4-\eta)\eta}{10}$, and $1-\Fid(\tilde{\Lambda}_\iSd, \Sgate^\dagger) = \frac{\eta}{2}$.

\begin{figure}[t!]
    \centering
\includegraphics[width=\linewidth]{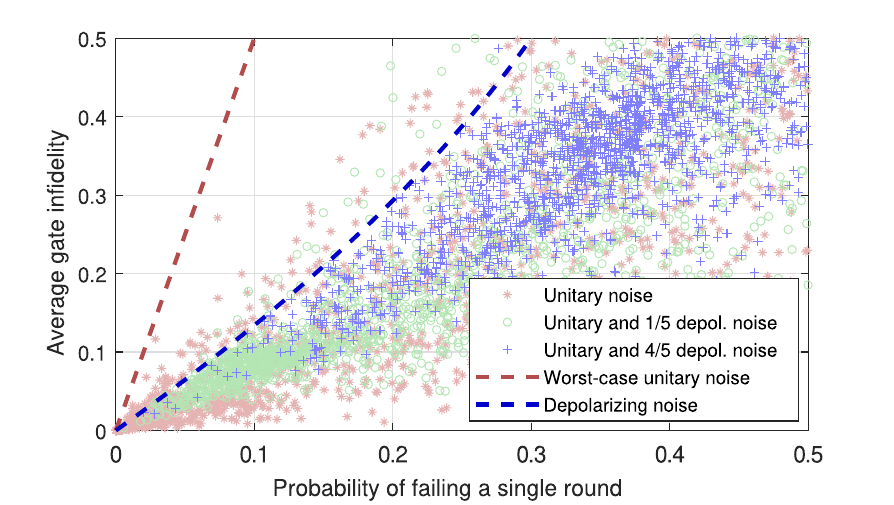}\\
\includegraphics[width=\linewidth]{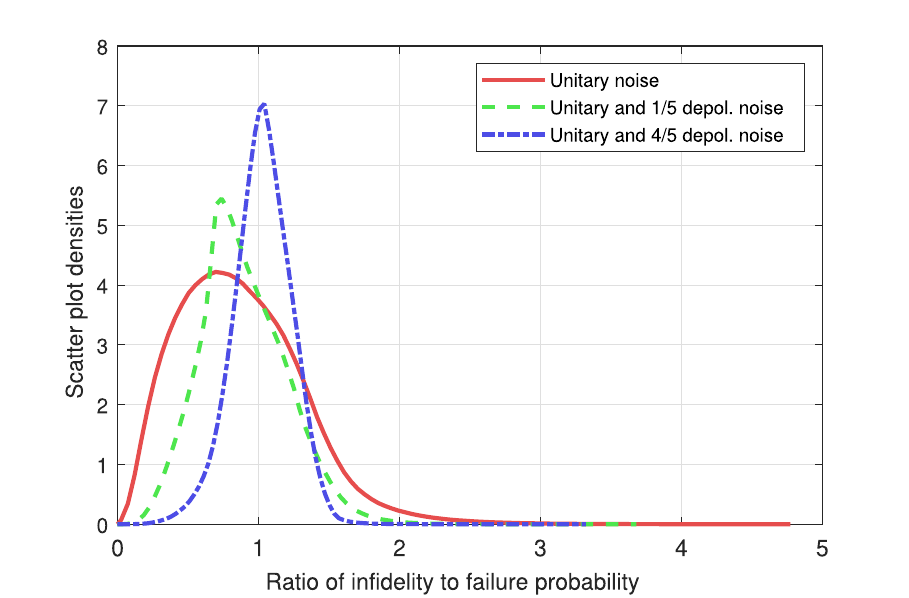}
\caption{Results of our numerical investigations. Top: dependence of the average gate infidelity on the probability of failing a single round of \cref{protocol} for noisy models, as in the statement of \cref{th:selftest_S}. Each of the approximately $10^3$ depicted points represents a single noisy model, with different noise models considered. 
The straight dotted line corresponds to the worst-case scaling among the randomly sampled noisy models, and the dotted curve corresponds to the depolarizing noise. Bottom: a collection of histograms representing distributions of the ratios of the average gate infidelity to the failure probability for the three considered cases of noise models. 
Each histogram is based on $10^7$ random samples. }
\label{fig:numerics}
\end{figure}

In \cref{fig:numerics} (top) we plot approximately $10^3$ randomly generated noisy quantum models, as well as a linear worst-case upper bound on the distance given the failing probability, which we estimated from sampling $10^7$ random models.
In addition to the unitary noise, we consider adding depolarizing noise with the depolarizing parameters $\frac{1}{5}$ and $\frac{4}{5}$, as well as the purely depolarizing noise added to the $\Sgate^\dagger$, as described above.
In \cref{fig:numerics} (bottom) we present histograms for the considered noise models, apart from the purely depolarizing noise. 
As one can see from the histograms, the unitary noise results in better scaling on average, but it is also the one which produces the worst-case behavior.

The performed numerical investigation follows the theoretical predictions of the linear scaling of the distances as a function of $\veps$ in \cref{th:selftest_S}, and allows us to estimate the coefficient of this linear function to be approximately $5$.
This results in the sample complexity of $\frac{5}{\veps}\ln(\delta^{-1})$, which amounts to approximately $2000$ repetitions for $\veps = \delta = 0.01$, or approximately $300$ repetitions for $\veps = \delta = 0.05$.

\subsection{Certification of a gate set universal for single-qubit quantum computation}
Next, we show how to employ \cref{protocol} to certify a universal gate set for single-qubit quantum computation. 
For this, we rely on already proven \cref{th:selftest_S} for self-testing of the $\Sgate$ gate, and incorporate the Hadamard gate $\Hgate$ and the $\Tgate$ gate to the sequences considered in the protocol.  
There are, however, important differences from the case of the $\Sgate$ gate certification.
First, in order to include the Hadamard we also need to account for a possible complex conjugation, which is still in accordance with \cref{def:correct_model}.
To certify the $\Tgate$ gate, we would either need to modify \cref{protocol} to include estimation of outcome probabilities, or as we do it here, change the goal of the certification. 
In particular, we show that using \cref{protocol}, we can certify implementation of the square root of the $\Sgate$ gate, which can be either $\Tgate = \kb{0}{0}+\e^{\i\frac{\pi}{4}}\kb{1}{1}$, or $\Zgate\Tgate = \kb{0}{0}-\e^{\i\frac{\pi}{4}}\kb{1}{1}$.
Simultaneously, either of the two gates, $\Tgate$ or $\Zgate\Tgate$,  in conjunction with the $\Sgate$ gate and the Hadamard gate, constitute a universal gate set for single-qubit quantum computation.

We use the following gate sequences for certification of the quantum model $\Universalmodel$, 
\begin{align}\label{eq:X_uni}
\begin{split}
    \mathcal{X}=\{&\epsilon,\iS\iSd,\iSd\iS,\iS\iS,\iSd\iSd,\iS\iH\iS,\iSd\iH\iS,\iH\iH,\iH\iS\iH,\\
     &\iH\iT\iH,\iS\iS\iH\iT\iH,\iT\iT\iS\},
\end{split}
\end{align}
where the labels $X = \Set{\iS,\iSd,\iH,\iT}$ correspond to the gates in the self-explanatory way.
Recall, that we read the sequences from left to right, e.g., in the sequence $\iSd\iH\iS$, the gate corresponding to $\iSd$ is applied first.
We again take $A = \Set{+,-}$, and set the deterministic outcomes expected by \cref{protocol} to 
\begin{equation}\label{eq:outcomes_uni}
a_\vec{x} = \left\{ 
  \begin{array}{ c l }
    + & \textrm{for } \vec{x}\in \Set{\epsilon,\iS\iSd,\iSd\iS,\iS\iH\iS,\iH\iH,\iH\iS\iH,\iH\iT\iH},\\
    - & \textrm{for } \vec{x}\in \Set{\iS\iS,\iSd\iSd,\iSd\iH\iS,\iS\iS\iH\iT\iH,\iT\iT\iS}.
  \end{array}
\right.
\end{equation}
We formalize our findings in this direction in the following theorem, which is an ideal self-testing type result.

\begin{theorem}\label{th:universal}
If a single-qubit quantum model $(\tilde{\rho},\{\tilde{\Lambda}_x\}_{x\in X},\{\tilde{M}_+,\tilde{M}_-\})$, with $X = \Set{\iS,\iSd,\iH,\iT}$ passes a single repetition of \cref{protocol} with probability $1$, for $\mathcal{X}$ and $a_\vec{x}$ specified in \cref{eq:X_uni} and \cref{eq:outcomes_uni}, respectively, and for any sampling distribution such that $\mu(\vec{x})>0$ for all $\vec{x}\in\mathcal{X}$, then
each quantum channel $\tilde{\Lambda}_x$ in the model is unitary, i.e., there exist $\tilde{U}_x\in\U(2)$, such that
\begin{equation}
    \tilde{\Lambda}_x(\argdot) = \tilde{U}_x(\argdot)\tilde{U}_x^\dagger,\quad \forall x\in X,
\end{equation}
and there exists a unitary $U\in\U(2)$ (with possible complex conjugation $^{(\ast)}$), such that
\begin{align}\label{eq:th_universal}
\begin{split}
    \tilde{U}_\iS &= U \Sgate^{(\ast)} U^\dagger,\\
    \tilde{U}_\iSd &= U {\Sgate^\dagger}^{(\ast)} U^\dagger,\\
    \tilde{U}_\iH &= U \Hgate U^\dagger,
\end{split}
\end{align}
and either $\tilde{U}_\iT = U \Tgate^{(\ast)} U^\dagger$ or $\tilde{U}_\iT = U \Zgate\Tgate^{(\ast)} U^\dagger$.
Moreover, for the same unitary $U$ it holds that,
\begin{align}\label{eq:th_uni_state_meas}
\begin{split}
    \tilde{\rho} &= U\kb{+}{+} U^\dagger,\quad \tilde{M}_+ = U \kb{+}{+} U^\dagger.\\
\end{split}
\end{align}
\end{theorem}

In the statement of \cref{th:universal}, we define the complex conjugation with respect to the computational basis.
\begin{proof}
We start with a brief overview of the proof steps. 
The starting point is to consider the sequences $\Set{\epsilon,\iS\iSd,\iSd\iS,\iS\iS,\iSd\iSd}$, which achieve the self-testing result for the $\Sgate$, and  $\Sgate^\dagger$ gates. 
By additionally considering all the sequences composed of $\Set{\iH,\iS,\iSd}$, we can self-test the Hadamard gate. 
Lastly, we consider sequences involving $\iT$ and prove that either the $\Tgate$ gate or $\Zgate\Tgate$ is implemented when we pass this instruction label to the quantum computer. 
    
The first step follows immediately from the proof of Theorem \ref{th:selftest_S} for the special case $\veps=0$, which provides us with a unitary $U$ such that the first two equations in \cref{eq:th_universal} as well as \cref{eq:th_uni_state_meas} are satisfied. 
Note, that at this point we do not need additional complex conjugation.
Let $\kb{\psi}{\psi}=\tilde{\rho}$ be the initial state and $\kb{\phi}{\phi}=\tilde{\Lambda}_\iSd(\kb{\psi}{\psi})$. 
Following \cref{th:selftest_S}, we know that $\tilde{\Lambda}_\iS$ and $\tilde{\Lambda}_\iSd$ implement a unitary $\tilde{U}_\iS$ and its inverse, respectively, where 
\begin{equation}\label{eq:proof_ideal_S_2}
     \tilde{U}_\iS=\kb{\psi}{\phi}+\kb{\psi^\perp}{\phi^\perp}.
\end{equation}
Moreover, following the proof of \cref{th:selftest_S}, we also obtain that the bases $\{\ket{\psi},\ket{\psi^\perp}\}$ and $\{\ket{\phi},\ket{\phi^\perp}\}$ are mutually unbiased, and we can choose all the 
inner products of $\Set{\ket{\psi},\ket{\psi^\perp},\ket{\phi},\ket{\phi^\perp}}$ to be real (see \cref{app:eq:proof_S_cond_2}).

We continue the proof with self-testing of the Hadamard gate. 
The observed correlations for the strings $\vec{x}=\iS\iH\iS$, $\vec{x}=\iSd\iH\iS$ imply that $\tilde{\Lambda}_\iH(\kb{\phi}{\phi}) = \kb{\phi^\perp}{\phi^\perp}$ and $\tilde{\Lambda}_\iH(\kb{\phi^\perp}{\phi^\perp}) = \kb{\phi}{\phi}$, i.e., $\tilde{\Lambda}_\iH$ maps an \ac{ONB} to an \ac{ONB}.
From the input string $\vec{x}=\iH\iH$, we also determine that $\tilde{\Lambda}_\iH\circ\tilde{\Lambda}_\iH(\kb{\psi}{\psi}) = \kb{\psi}{\psi}$, which implies that $\tilde{\Lambda}_\iH(\kb{\psi}{\psi})$ is a pure state (see \cref{lemma:double_channel_purity}).
Combining these, we then conclude that $\tilde{\Lambda}_\iH$ is a unitary channel (see \cref{app:lemma_unitarity}).
Moreover, the corresponding unitary operator $\tilde{U}_\iH$ must be of the form 
\begin{equation}
    \tilde{U}_\iH=\kb{\phi}{\phi^\perp}+\e^{\i\theta}\kb{\phi^\perp}{\phi}
\end{equation}
for some $\theta\in\RR$. 

We can identify the phase $\theta$ by considering the sequence input string $\vec{x}=\iH\iS\iH$. 
Specifically, the observed deterministic behavior implies that $\abs{\bra{\psi}\tilde{U}_\iH \tilde{U}_\iS \tilde{U}_\iH\ket{\psi}}=1$, with $\tilde{U}_\iS$ as in \cref{eq:proof_ideal_S_2}.
This results in the condition
\begin{equation}
 \abs{1+2\e^{\i\theta}-\e^{2\i\theta}}=2\sqrt{2}, 
\end{equation}
which is satisfied if and only if $\e^{\i\theta}=\pm \i$. 
By applying the gauge unitary $U$, this leaves either of the two possibilities $U^\dagger \tilde{U}_\iH U\in\{\Hgate ,\Xgate\Hgate \Xgate\}$ (up to a global phase). 
In the latter case, we absorb the $\Xgate$ gate into the gauge unitary $U\mapsto U\Xgate$, at the cost of interchanging $\Sgate$ and $\Sgate ^\dagger$ (and an additional global phase), which effectively amounts to applying a complex conjugation to $\Sgate$ and $\Sgate^\dagger$.
Note, that adding $\Xgate$ to the gauge unitary does not change the results for the state and the measurement in \cref{eq:th_uni_state_meas} since $\ket{+},\ket{-}$ are eigenstates of $\Xgate$. 

Finally, we also consider the sequences involving the `$\iT$' input. 
The observed correlations for the input strings $\vec{x}=\iH\iT\iH$ and $\vec{x}=\iS\iS\iH\iT\iH$ imply that the channel $\tilde{\Lambda}_\iT$ maps an \ac{ONB} $\Set{\tilde{U}_\iH\ket{\psi},\tilde{U}_\iH\ket{\psip}}$ to an \ac{ONB} (in fact, to itself).
Moreover, from the sequence $\vec{x}=\iT\iT\iS$ we can deduce that $\tilde{\Lambda}_\iT(\psi)$ is a pure state (see \cref{lemma:double_channel_purity}). 
Moreover, since $\abs{\bra{\psi}\tilde{U}_\iH\ket{\psi}}=1/\sqrt{2}$, we can invoke \cref{app:lemma_unitarity} to conclude that $\tilde{\Lambda}_\iT$ is a unitary channel. 
From $\abs{\bra{\psi} \tilde{U}_\iH \tilde{U}_\iT \tilde{U}_\iH\ket{\psi}}=1$ we deduce, that after applying the gauge unitary (and a possible complex conjugation), we have 
\begin{equation}
    U^\dagger \tilde{U}_\iT U=\kb{0}{0}+\e^{\i\varphi}\kb{1}{1},
\end{equation}
for a suitable phase $\varphi\in \RR$. 
This phase is constrained by the input string $\vec{x} = \iT\iT\iS$, since $\abs{\bra{\psi^\perp} \tilde{U}_\iS \tilde{U}_\iT \tilde{U}_\iT\ket{\psi}}=1$ is equivalent to $\e^{2\i\varphi}=\i$, which leaves the two possibilities $\varphi\in \Set{\frac{\pi}{4},\pi+\frac{\pi}{4}}$,
which we cannot distinguish further with \cref{protocol}, but which lets us to conclude that $\tilde{U}_\iT \in \Set{U\Tgate^{(\ast)}U^\dagger,U\Zgate\Tgate^{(\ast)}U^\dagger}$ (up to the global phase).
This finishes the proof.
\end{proof}

It is possible to modify \cref{protocol} for self-testing the $\Tgate$ gate in the sense of \cref{def:correct_model},
by including sequences such as, e.g., $\vec{x} = \iT$.
However, this means that the target statistics will stop being deterministic, and we will need to estimate the corresponding outcome probabilities up to some precision.
At the same time, it does not mean that the overall sampling complexity should change drastically, because, at least in the ideal case, we will only need to distinguish between the two cases $\abs{\bra{+}\Tgate\ket{+}}^2 = \frac{1}{2}+\frac{1}{2\sqrt{2}}$, and $\abs{\bra{+}\Zgate\Tgate\ket{+}}^2 = \frac{1}{2}-\frac{1}{2\sqrt{2}}$.

\section{Conclusions and outlook}
In this paper, we propose a novel method for certifying quantum gates within a practical scenario where a classical user interacts with a quantum server, considered as a black box.
The method certifies the target gates together with the state preparation and measurement, thus making it free from \ac{SPAM} errors.
Here, we focus on single-qubit gates and prove soundness of certification for a gate set that is universal for single-qubit quantum computation, based on a few minimal assumptions.  
Moreover, for a relevant single-qubit phase gate, which corresponds experimentally to a $\pi/2$-pulse, we show that the sample complexity of our method scales like $\LandauO(\veps^{-1})$ with respect to the average gate infidelity $\veps$.

Within the range of quantum system characterization techniques, the proposed method occupies a unique position and cannot be substituted by any existing tools.
While direct gate certification and tomography are affected by \ac{SPAM} errors, \ac{SPAM}-robust characterization methods are not supported by the soundness guarantees if applied for certification.   
Self-testing either requires two isolated parts of an experimental apparatus or computational assumptions, both posing challenges in a way of it being applied in practice at the current level of technological development of quantum hardware.
This work presents a fresh perspective that can inspire the development of practical and reliable certification techniques for testing quantum computers.

In this work, we focus on certification of single-qubit gates. 
Nevertheless, some of the introduced concepts have the potential for extension.
\Cref{th:selftest_S}, can be used to construct a \emph{fidelity witness} for quantum gates, if \cref{protocol} is modified to estimate the acceptance probability. 
This will be particularly relevant for an experimental demonstration of the proposed method.
The soundness proof of \cref{th:universal} can be extended to multi-qubit quantum gates, as our preliminary analysis suggests.
However, this extension requires new techniques for the case of entangling gates, and, thus, deserves a separate study.

It also seems possible to translate some of the ideas from Ref.~\cite{van2000self} to the framework of the dimension assumption, removing the requirement on the ideal preparation of the computational basis states, assumed therein.
Finally, we find the connection between the classical simulability of quantum computation and the types of quantum gates which can be efficiency certified with deterministic measurement outcomes intriguing, which also deserves a separate investigation. 

\begin{acknowledgments}
We thank Micha\l{} Oszmaniec, Costantino Budroni, and Ingo Roth for inspiring discussions. 
This research was funded by the Deutsche Forschungsgemeinschaft (DFG, German Research Foundation), project numbers 441423094, 236615297 - SFB 1119) and the Fujitsu Services GmbH as part of the endowed professorship ``Quantum Inspired and Quantum Optimization''.
\Cref{fig:user-server} is drawn using tikzpeople package, developed by Nils Fleischhacker.
\Cref{fig:numerics} is generated in MATLAB.
\end{acknowledgments}

%\newpage
\onecolumngrid
\section*{Appendix}
\begin{appendix}
\section{Proof of Theorem~\ref{th:selftest_S}}\label{app:th_selftest_S}
We repeat the statement of the theorem for convenience. 
We omit ``tilde'' over the implemented state, channels, and the measurement to keep the presentation simple. 
\begin{theorem1}
Let $(\rho,\{\Lambda_\iS,\Lambda_\iSd\},\{M_+,M_-\})$ be a single-qubit quantum model that passes with probability at least $1-\veps$ a single repetition of \cref{protocol} with uniform sampling from the gate sequences \eqref{eq:X_S} and for deterministic measurement outcomes \eqref{eq:outcomes_S}. 
Then the quantum model is $\LandauO(\veps)$-close to the target model \eqref{eq:S_gate_target_model}, i.e.,
there is a unitary $U\in\U(2)$ such that
\begin{align}\label{app:eq:th_selftest_S}
\begin{split}
\Fid(\Lambda_\iS, U\Sgate U^\dagger) &\geq 1-\LandauO(\veps),\\
\Fid(\Lambda_\iSd, U\Sgate^\dagger U^\dagger) &\geq 1-\LandauO(\veps)\,,\\
\Tr[\rho U\kb{+}{+}U^\dagger] &\geq 1-\frac{15}{2}\veps,\\
\norm*{M_{+}-U\kb{+}{+}U^\dagger}_\infty &\leq \frac{5}{2}\veps\, .
\end{split}
\end{align}
\end{theorem1}
\begin{proof}
Because the following proof is lengthy and technical in parts, we start by giving a general outline. 
The conclusions of the theorem follow from the condition $\Pass\geq 1-\veps$ and the dimension assumption, that is $\rho,M_+,M_-\in\L(\CC^2)$, and $\Lambda_\iS:\L(\CC^2)\to\L(\CC^2)$, $\Lambda_\iSd:\L(\CC^2)\to\L(\CC^2)$. 
As a first step, we show that for small $\veps$, the measurement effects ${M}_+$ and ${M}_-$ are close to being rank-1 projectors, which we denote as $\psi$ and $\psip$.
Next, we show that \acp{POVM} which one obtains by applying the adjoint maps $\Lambda_\iS^\dagger$ and $\Lambda_\iSd^\dagger$ to $\psi$ and $\psip$ are also close to be projective for small $\veps$.
We denote the corresponding projectors by $\phi$ and $\phip$.
Importantly, we find that $\Lambda_\iS^\dagger(\psi)\approx\phi$ and $\Lambda_\iSd^\dagger(\psi)\approx\phip$, and since the adjoint maps of channels are unital, also $\Lambda_\iS^\dagger(\psip)\approx\phip$ and $\Lambda_\iSd^\dagger(\psip)\approx\phi$.  
Next, we obtain a partial characterization of the Choi state of the channel $\Lambda_\iS$ in the basis of $\psi$ and $\phi$, with the leading terms which we denote as $a_1,a_2,a_3,a_3^\ast$ corresponding to the subspace spanned by $\ket\psi\ket\phi$ and $\ket\psip\ket\phip$.
Here, $a_3$ and $a_3^\ast$ correspond to the off-diagonal terms of the matrix representation of $\Choi(\Lambda_\iS)$, which at this point can only be upper-bounded by $\frac{1}{2}$.
The case $\abs{a_3}\approx \frac{1}{2}$ corresponds to $\Lambda_\iS$ being a unitary channel. 
In order to show that actually $a_3\approx \frac{1}{2}$, we use the conditions $\rho\approx \psi$ and $\Lambda_\iS(\rho)\approx\phip$. 
Finally, we find a suitable gauge unitary $U\in\U(2)$ for which the condition for $\Lambda_\iS$ in \cref{app:eq:th_selftest_S} follows.
Because we obtain characterization of $\Lambda_\iS$ and $\Lambda_\iSd$ in the same basis, the proof also easily extends to the channel $\Lambda_\iSd$.
The bounds for the state and the measurement in \cref{app:eq:th_selftest_S} for the chosen unitary are also immediate.
Showing each step of the above sketch is, in principle, not too technical, but
a lot of involving calculations in the proof are there to ensure the linear scaling of the bounds in \cref{app:eq:th_selftest_S} with respect to $\veps$.

We start the proof by writing the probability of a quantum model, given by $\rho$, $\Lambda_\iS$, $\Lambda_\iSd$, and $\{M_+,M_-\}$ passing a single repetition of the protocol.
\begin{equation}\label{app:eq:pass_S}
\begin{split}
\Pass & = \frac{1}{5}\Big(\tr[M_+\rho]+\tr[M_+\Lambda_\iS\circ\Lambda_\iSd(\rho)]+\tr[M_+\Lambda_\iSd\circ\Lambda_\iS(\rho)]\\
& +\tr[M_-\Lambda_\iS\circ\Lambda_\iS(\rho)]+\tr[M_-\Lambda_\iSd\circ\Lambda_\iSd(\rho)]\Big).
\end{split}
\end{equation}
For simplicity, let us take $\veps$ such that $\Pass\geq 1-\frac{\veps}{5}$, and rescale $\veps$ at the end of the proof.
We separate the condition in \cref{app:eq:pass_S} into $\Tr[M_+\rho]\geq 1-\veps$ and
\begin{equation}\label{app:eq:proof_S_1}
\tr[M_+\Lambda_\iS\circ\Lambda_\iSd(\rho)]+\tr[M_+\Lambda_\iSd\circ\Lambda_\iS(\rho)]-\tr[M_+\Lambda_\iS\circ\Lambda_\iS(\rho)]-\tr[M_+\Lambda_\iSd\circ\Lambda_\iSd(\rho)]\geq 2-\veps.
\end{equation}
Let the eigendecomposition of $M_+$ be $ M_+ = (1-\lambda_+)\psi+\lambda_-\psip$, where $\psi\coloneqq\kb{\psi}{\psi}$, $\psip\coloneqq\kb{\psip}{\psip}$, $\bk{\psi}{\psip} = 0$, and $\lambda_++\lambda_-\leq 1$.
We can then substitute $M_+$ in \cref{app:eq:proof_S_1} with $(1-\lambda_+-\lambda_-)\psi+\lambda_-\1$, and due to the normalization of states and $1\geq (1-\lambda_+-\lambda_-)$, we arrive at the same condition as  \cref{app:eq:proof_S_1}, but with $\psi$ instead of $M_+$.
We also obtain that $\lambda_++\lambda_-\leq\frac{\veps}{2}$, because we can upper-bound the expression, which is multiplied by $(1-\lambda_+-\lambda_-)$ on the left-hand side of \cref{app:eq:proof_S_1} by $2$.  
Next, for each trace, we move the second channel in the sequence to the measurement side, and denote the adjoint maps as $\Lambda_\iS^\dagger$ and $\Lambda_\iSd^\dagger$.
Grouping the terms together, we obtain
\begin{equation}\label{app:eq:proof_S_2}
\tr\left[\left(\Lambda_\iS^\dagger(\psi)-\Lambda_\iSd^\dagger(\psi)\right)\Big(\Lambda_\iSd(\rho)-\Lambda_\iS(\rho)\Big)\right]\geq 2-\veps.
\end{equation}

Let $\Lambda_\iS^\dagger(\psi)-\Lambda_\iSd^\dagger(\psi) = \eta_+\phi-\eta_-\phip$, where $\phi\coloneqq\kb{\phi}{\phi}$, $\phip\coloneqq\kb{\phip}{\phip}$, $\bk{\phi}{\phip} = 0$, and $\eta_+,\eta_-\in [-1,1]$ due to the fact that \ac{POVM} effects are \ac{PSD} and bounded.   
Inserting this eigendecomposition into \cref{app:eq:proof_S_2}, leads to
\begin{equation}\label{app:eq:proof_S_3}
(\eta_++\eta_-)\tr\left[\phi\Big(\Lambda_\iSd(\rho)-\Lambda_\iS(\rho)\Big)\right]\geq 2-\veps.
\end{equation}
Since the trace in \cref{app:eq:proof_S_3} can be at most $1$, and each of $\eta_-$ and $\eta_+$ are upper-bounded by $1$, we conclude that $\eta_-\geq 1-\veps$ and $\eta_+\geq 1-\veps$.
From this conclusion, we arrive at a first set of important conditions that characterize the channels $\Lambda_\iS$ and $\Lambda_\iSd$, namely
\begin{equation}\label{app:eq:proof_S_cond_1}
\tr[\Lambda_\iS^\dagger(\psi)\phi]\geq 1-\veps,\quad \tr[\Lambda_\iS^\dagger(\psip)\phip]\geq 1-\veps,\quad \tr[\Lambda_\iSd^\dagger(\psi)\phip]\geq 1-\veps,\quad \tr[\Lambda_\iSd^\dagger(\psip)\phi]\geq 1-\veps.
\end{equation}

Next, we focus on channel $\Lambda_\iS$ and derive a partial characterization of its Choi state.
We define the Choi-Jamio\l{}kowski state~\cite{Cho75,Jam72}, or the Choi state for short, of a qubit channel $\Lambda$ and the inverse Choi map with respect to the canonical product basis $(\ket{i}\ket{j})_{i,j\in\Set{0,1}}$ in $\CC^4$ as
\begin{equation}
\Choi(\Lambda) \coloneqq \frac{1}{2}\sum_{i,j\in\Set{0,1}}\Lambda(\kb{i}{j})\otimes\kb{i}{j},\quad \Lambda^\dagger(\argdot) = 2\left(\tr_1[(\argdot)\otimes\1\Choi(\Lambda)]\right)^\T,
\end{equation}
where $\tr_1[\argdot]$ denotes the partial trace with respect to the first subsystem. 
Let us specify the matrix representation of $\Choi(\Lambda_\iS)$ in the basis ${\mathrm{ONB}_1}\coloneqq (\ket\psi\ket\phi^\ast,\ket\psip\ket\phip^\ast,\ket\psi\ket\phip^\ast,\ket\psip\ket\phi^\ast)$ as follows
\begin{equation}\label{app:eq:proof_S_char}
\left[\Choi(\Lambda_\iS)\right]_{\mathrm{ONB}_1}=
\begin{bmatrix}
A & B \\
B^\dagger & C 
\end{bmatrix}
\coloneqq
\begin{bmatrix}
	a_1 & a_3 & b_1 & b_2\\
	a_3^\ast & a_2 & b_3^\ast & -b_1^\ast\\
	b_1^\ast & b_3 & c_2 & c_3\\
	b_2^\ast & -b_1 & c_3^\ast & c_1
\end{bmatrix}, 
\end{equation}
where $A,B,C\in\CC^{2\times 2}$ represent the $2\times 2$ blocks of $\left[\Choi(\Lambda_\iS)\right]_{\mathrm{ONB}_1}$, and $a_1,a_2,c_1,c_2\in\RR$, and $a_3,b_1,b_2,b_3,c_3\in \CC$ represent the entries.
From the derived condition in \cref{app:eq:proof_S_cond_1}, we have that $a_1\geq \frac{1}{2}-\frac{\veps}{2}$ and $a_2\geq \frac{1}{2}-\frac{\veps}{2}$.
From the normalization condition $\Tr_1[\Choi(\Lambda_s)] = \frac{\1}{2}$, we have that $c_1=\frac{1}{2}-a_1$ and $c_2=\frac{1}{2}-a_2$, and, therefore, $c_1\leq \frac{\veps}{2}$ and $c_2\leq \frac{\veps}{2}$.
From the \ac{PSD} condition $\left[\Choi(\Lambda_\iS)\right]_{\mathrm{ONB}_1}\geq 0$, we obtain that 
\begin{equation}\label{app:eq:proof_S_upper_bounds_entries}
\abs{a_3}\leq \frac{1}{2},\quad \abs{c_3}\leq \frac{\veps}{2},\quad \abs{b_1}\leq \frac{\sqrt{\veps}}{2},\quad \abs{b_2}\leq \frac{\sqrt{\veps}}{2},\quad \abs{b_3}\leq \frac{\sqrt{\veps}}{2}.
\end{equation}
We can use the above estimates to upper-bound the unwanted terms, i.e., all except for the ones in submatrix $A$, in channel $\Lambda_\iS$.  
However, they are not sufficient for obtaining the linear scaling in $\veps$ of the bounds in \cref{app:eq:th_selftest_S}.
We will also need a tighter upper-bound on $\abs{b_2+b_3}$.

In order to derive a tighter upper-bound on $\abs{b_2+b_3}$, we use the following constraint on the blocks $A,B,C$ that form a \ac{PSD} matrix,
\begin{equation}\label{app:eq:Horn_ineq}
\abs{\bra{v}B\ket{w}}^2\leq \bra{v}A\ket{v}\bra{w}C\ket{w}, \quad \forall \ket{v},\ket{w}\in \CC^2. 
\end{equation}
This result can be found in Ref.~\cite{horn1985matrix} (Theorem 7.7.7), and we also provide a proof of \cref{app:eq:Horn_ineq} in \cref{app:lemmata} for completeness.  
Let us first take $\ket{v} = -\frac{a^{\ast}_3}{\abs{a_3}}\ket{\phip}^\ast+\ket{\phi}^\ast$ and $\ket{w} = \ket{\phip}^\ast$.
The condition in \cref{app:eq:Horn_ineq} then implies
\begin{equation}
\abs*{-\frac{a_3}{\abs{a_3}}b_1^\ast+b_3}^2\leq (a_1+a_2-2\abs{a_3})c_2\leq (1-2\abs{a_3})\frac{\veps}{2}. 
\end{equation}
Next, take $\ket{v} = \ket{\phip}^\ast-\frac{a_3}{\abs{a_3}}\ket{\phi}^\ast$ and $\ket{w} = \ket{\phi}^\ast$, which results in a similar condition,
\begin{equation}
\abs*{b_2+\frac{a_3}{\abs{a_3}}b_1^\ast}^2\leq (a_1+a_2-2\abs{a_3})c_1\leq (1-2\abs{a_3})\frac{\veps}{2}. 
\end{equation}
Using the triangular inequality, we then obtain a condition 
\begin{equation}\label{app:eq:proof_S_upper_bound_b2b3}
\abs{b_2+b_3}\leq \sqrt{2\veps}\sqrt{1-2\abs{a_3}},
\end{equation}
which we use later in the proof. 

We continue the proof by returning to \cref{app:eq:proof_S_3} and using the condition $\eta_++\eta_-\leq 2$, obtain that $\Tr[\phi\Lambda_\iSd(\rho)]\geq 1-\frac{\veps}{2}$ and $\Tr[\phi\Lambda_\iS(\rho)]\leq \frac{\veps}{2}$.
Again, we focus on channel $\Lambda_\iS$ first, and rewrite the aforementioned condition for it as
\begin{equation}\label{app:eq:proof_S_cond_2}
\Tr[\Lambda^\dagger_\iS(\phip)\rho] \geq 1-\frac{\veps}{2}.
\end{equation}
This is the second important condition alongside \cref{app:eq:proof_S_cond_1} that allows us to characterize channel $\Lambda_\iS$.

It is useful at this point of the proof to fix the relative phases between the vectors $\ket{\phi},\ket{\phip},\ket{\psi}$, and $\ket{\psip}$.
Without loss of generality, we set
\begin{equation}\label{app:eq:overlaps_convention}
\bk{\psi}{\phi} = \bk{\psip}{\phip} = \abs{\bk{\psi}{\phi}}, \quad -\bk{\psi}{\phip} = \bk{\psip}{\phi} = \abs{\bk{\psip}{\phi}}.
\end{equation}
Let us first express $\rho$ in the basis $\Set{\ket{\psi},\ket{\psip}}$ as 
\begin{equation}\label{app:eq:proof_S_4}
\rho = d_1\psi + d_2\psip+d_3\kb{\psi}{\psip}+d^\ast_3\kb{\psip}{\psi}.
\end{equation}
From the condition $\Tr[M_+\rho]\geq 1-\veps$, which we obtained directly from \cref{app:eq:pass_S}, and from the condition on the eigenvalues of $M_+$, namely, $\lambda_++\lambda_-\leq \frac{\veps}{2}$, we obtain that $d_1 = \Tr[\psi\rho] \geq 1-\frac{3}{2}\veps$, and, consequently, $d_2\leq \frac{3}{2}\veps$.
From $\rho\geq 0$, we obtain additionally that $\abs{d_3}\leq \LandauO(\sqrt{\veps})$.

From now on, we express the bounds using the Big-O notation, because we are interested in the scaling w.r.t.~$\veps$, and we estimate the constants in our numerical studies in \cref{sec:results}. 
Using the expansion in \cref{app:eq:proof_S_4}, we can reduce the condition in \cref{app:eq:proof_S_cond_2} to
\begin{equation}\label{app:eq:proof_S_5}
\Tr[\Lambda^\dagger_\iS(\phip)\psi] + 2\Re\left[d_3\Tr[\Lambda^\dagger_\iS(\phip)\kb{\psi}{\psip}]\right]\geq 1-\LandauO(\veps).
\end{equation}
We do not simply use the upper bound of $\LandauO(\sqrt{\veps})$ on the second term in \cref{app:eq:proof_S_5}, but instead carefully analyze both terms. 
We use the expansion of $\phip$ in the basis of $(\ket{\psi},\ket{\psip})$ to write the \ac{POVM} effect $\Lambda^\dagger_\iS(\phip)$ as
\begin{equation}\label{app:eq:proof_S_6}
\Lambda^\dagger_\iS(\phip) = \abs{\bk{\psip}{\phi}}^2\Lambda^\dagger_\iS(\psi)+\abs{\bk{\psi}{\phi}}^2\Lambda^\dagger_\iS(\psip) - \abs{\bk{\psi}{\phi}\bk{\psip}{\phi}}\Big(\Lambda^\dagger_\iS(\kb{\psi}{\psip})+\Lambda^\dagger_\iS(\kb{\psip}{\psi})\Big).
\end{equation}
We can use \cref{app:eq:proof_S_6} and the partial characterization of $\Lambda_\iS$ in \cref{app:eq:proof_S_char} to express $\Lambda^\dagger_\iS(\phip)$ in the basis ${\mathrm{ONB}_2}\coloneqq (\ket{\phi},\ket{\phip})$,
\begin{align}\label{app:eq:proof_S_7}
[\Lambda^\dagger_\iS(\phip)]_{\mathrm{ONB}_2} &= 2
\begin{bmatrix}
	c_1\abs{\bk{\psi}{\phi}}^2 & -c_3\abs{\bk{\psi}{\phi}\bk{\psip}{\phi}}\\
	-c_3^\ast\abs{\bk{\psi}{\phi}\bk{\psip}{\phi}} & c_2\abs{\bk{\psip}{\phi}}^2
\end{bmatrix}\\
&+2
\begin{bmatrix}
	a_1\abs{\bk{\psip}{\phi}}^2-2\Re[b_2]\abs{\bk{\psi}{\phi}\bk{\psip}{\phi}} & b_1^\ast(\abs{\bk{\psip}{\phi}}^2-\abs{\bk{\psi}{\phi}}^2)-a_3^\ast\abs{\bk{\psi}{\phi}\bk{\psip}{\phi}}\\
	b_1(\abs{\bk{\psip}{\phi}}^2-\abs{\bk{\psi}{\phi}}^2)-a_3\abs{\bk{\psi}{\phi}\bk{\psip}{\phi}} & a_2\abs{\bk{\psi}{\phi}}^2-2\Re[b_3]\abs{\bk{\psi}{\phi}\bk{\psip}{\phi}}
\end{bmatrix}.\nonumber
\end{align}
The first summand in the above expression can be safely ignored, because its contribution is of the order of $\LandauO(\veps)$, due to the upper-bounds on its entries.
On the other hand, the matrix representations of $\psi$ and $\kb{\psi}{\psip}$ in the basis ${\mathrm{ONB}_2}$, are
\begin{align}\label{app:eq:proof_S_8}
\begin{split}
[\psi]_{\mathrm{ONB}_2} &= 
\begin{bmatrix}
	\abs{\bk{\psi}{\phi}}^2 & -\abs{\bk{\psi}{\phi}\bk{\psip}{\phi}}\\
	-\abs{\bk{\psi}{\phi}\bk{\psip}{\phi}} & \abs{\bk{\psip}{\phi}}^2
\end{bmatrix}, \\
[\kb{\psi}{\psip}]_{\mathrm{ONB}_2} &= 
\begin{bmatrix}
	\abs{\bk{\psi}{\phi}\bk{\psip}{\phi}} & \abs{\bk{\psi}{\phi}}^2 &\\
	-\abs{\bk{\psip}{\phi}}^2 & -\abs{\bk{\psi}{\phi}\bk{\psip}{\phi}}
\end{bmatrix}.
\end{split}
\end{align}

Using \cref{app:eq:proof_S_7} and \cref{app:eq:proof_S_8}, we can upper-bound the first term on the left-hand side of \cref{app:eq:proof_S_5} as
\begin{align}\label{app:eq:proof_S_9}
\begin{split}
\Tr[\Lambda^\dagger_\iS(\phip)\psi] &\leq \LandauO(\veps)+2(a_1+a_2+2\Re[a_3])\abs{\bk{\psi}{\phi}\bk{\psip}{\phi}}^2\\
&-4\abs{\bk{\psi}{\phi}\bk{\psip}{\phi}}\Big(\abs{\bk{\psi}{\phi}}^2\Re[b_2]+\abs{\bk{\psip}{\phi}}^2\Re[b_3]+(\abs{\bk{\psip}{\phi}}^2-\abs{\bk{\psi}{\phi}}^2)\Re[b_1]\Big).
\end{split}
\end{align}
Similarly, the second term on the left-hand side of \cref{app:eq:proof_S_5} can be upper-bounded as
\begin{align}\label{app:eq:proof_S_10}
\begin{split}
\Re\left[d_3\Tr[\Lambda^\dagger_\iS(\phip)\kb{\psi}{\psip}]\right] \leq \LandauO(\veps)+2\Re\Big[d_3\abs{\bk{\psi}{\phi}\bk{\psip}{\phi}}\Big((a_1+a_3^\ast)\abs{\bk{\psip}{\phi}}^2-(a_2+a_3)\abs{\bk{\psi}{\phi}}^2\\
+2\abs{\bk{\psi}{\phi}\bk{\psip}{\phi}}\Re[(b_3-b_2)]\Big)+d_3\Big(\abs{\bk{\psip}{\phi}}^2-\abs{\bk{\psi}{\phi}}^2\Big)\Big(b_1\abs{\bk{\psi}{\phi}}^2-b_1^\ast\abs{\bk{\psip}{\phi}}^2\Big)\Big].
\end{split}
\end{align}
To simplify the estimates of the quantities in \cref{app:eq:proof_S_9} and \cref{app:eq:proof_S_10}, we introduce the last bit of notation, namely two functions $f:[0,1]\to[-1,1]$ and $g:[0,1]\to[0,2]$, such that
\begin{equation}\label{app:eq:proof_S_11}
\abs{\bk{\psi}{\phi}}^2 = \frac{1-f(\veps)}{2},\quad \abs{\bk{\psip}{\phi}}^2 = \frac{1+f(\veps)}{2},\quad \Re[a_3] = \frac{1-g(\veps)}{2}.
\end{equation}
Note, that even though we use $\veps$ as the argument for functions $f$ and $g$, there is no loss of generality in making the above assignments.
In particular, we can take $g(\veps)\geq 0$, because from \cref{app:eq:proof_S_upper_bounds_entries}, we know that $\abs{a_3}\leq \frac{1}{2}$.
Due to the same reason, we can upper-bound the absolute value of the imaginary part of $a_3$ as $\abs{\Im[a_3]}\leq \sqrt{\frac{g(\veps)}{2}}$.

Using the new notations in \cref{app:eq:proof_S_11}, as well as the upper bounds in \cref{app:eq:proof_S_upper_bounds_entries} and \cref{app:eq:proof_S_upper_bound_b2b3}, we can simplify the bound in \cref{app:eq:proof_S_9} as
\begin{equation}
\Tr[\Lambda^\dagger_\iS(\phip)\psi]\leq \frac{1}{2}(1-f(\veps)^2)(2-g(\veps))+\sqrt{1-f(\veps)^2}\sqrt{g(\veps)}\sqrt{2\veps}+2\abs{f(\veps)}\sqrt{\veps}+\LandauO(\veps).
\end{equation} 
Similarly, the bound in \cref{app:eq:proof_S_10} can be simplified as
\begin{equation}
\Re\left[d_3\Tr[\Lambda^\dagger_\iS(\phip)\kb{\psi}{\psip}]\right] \leq \LandauO(\sqrt{\veps})\sqrt{1-f(\veps)^2}\left(\abs{f(\veps)}+\sqrt{g(\veps)}\right)+\LandauO(\veps).
\end{equation}
Combining these two bounds together and inserting them back to the condition in \cref{app:eq:proof_S_5}, we finally arrive at
\begin{equation}
\left(\abs{f(\veps)}-\LandauO(\sqrt{\veps})\right)^2+\frac{1}{2}\left(\sqrt{1-f(\veps)^2}\sqrt{g(\veps)}-\LandauO(\sqrt{\veps})\right)^2\leq \LandauO(\veps).
\end{equation}
This allows us to deduce that $\abs{f(\veps)}\leq \LandauO(\sqrt{\veps})$ and $g(\veps)\leq \LandauO(\veps)$.

As the final part of the proof, we choose the gauge unitary $U$ to be
\begin{equation}\label{app:eq:proof_S_U}
U = \kb{\psi}{+}-\i\kb{\psip}{-}.
\end{equation}
Up to this gauge, the ideal gate $\Sgate$ takes the form
\begin{equation}\label{app:eq_ideal_S}
U\Sgate U^\dagger = \frac{\e^{\i\frac{\pi}{4}}}{\sqrt{2}}\left(\psi+\psip+\kb{\psi}{\psip}-\kb{\psip}{\psi}\right).
\end{equation}
Consequently, we find the Choi state vector $\ket{\Choi(U\Sgate U^\dagger)}$, which we define as $\kb{\Choi(U\Sgate U^\dagger)}{\Choi(U\Sgate U^\dagger)}\coloneqq \Choi(U\Sgate U^\dagger)$,
\begin{align}
\ket{\Choi(U\Sgate U^\dagger)} &= \frac{1}{\sqrt{2}}U\Sgate U^\dagger\otimes\1\left(\ket{0}\ket{0}+\ket{1}\ket{1}\right)\\
&=\frac{\e^{\i\frac{\pi}{4}}}{2}\Big((\abs{\bk{\psi}{\phi}}+\abs{\bk{\psip}{\phi}})(\ket{\psi}\ket{\phi}^\ast+\ket{\psip}\ket{\phip}^\ast)+(\abs{\bk{\psi}{\phi}}-\abs{\bk{\psip}{\phi}})(\ket{\psi}\ket{\phip}^\ast-\ket{\psip}\ket{\phi}^\ast)\Big).\nonumber
\end{align}
Its matrix representation in the $\mathrm{ONB}_1$ is 
\begin{equation}\label{app:eq:proof_S_ideal_Choi}
\left[\ket{\Choi(U\Sgate U^\dagger)}\right]_{\mathrm{ONB}_1} = \frac{\e^{\i\frac{\pi}{4}}}{2}
\begin{bmatrix}
\abs{\bk{\psi}{\phi}}+\abs{\bk{\psip}{\phi}}\\
\abs{\bk{\psi}{\phi}}+\abs{\bk{\psip}{\phi}}\\
\abs{\bk{\psi}{\phi}}-\abs{\bk{\psip}{\phi}}\\
-(\abs{\bk{\psi}{\phi}}-\abs{\bk{\psip}{\phi}})
\end{bmatrix}.
\end{equation}
Having the explicit forms of the Choi states in \cref{app:eq:proof_S_char} and \cref{app:eq:proof_S_ideal_Choi} allows us to estimate their inner product,
\begin{align}\label{app:eq:proof_S_final}
\tr\left[\Choi(\Lambda_\iS)\Choi(U\Sgate U^\dagger)\right] &= \frac{1}{4}\Big((\abs{\bk{\psi}{\phi}}+\abs{\bk{\psip}{\phi}})^2(a_1+a_2+2\Re[a_3])\\
&+(\abs{\bk{\psi}{\phi}}^2-\abs{\bk{\psip}{\phi}}^2)(2\Re[2b_1-b_2+b_3)\Big)+\LandauO(\veps)\\
&\geq \frac{1}{4}\Big((1+\sqrt{1-f(\veps)^2})(1-\veps+1-g(\veps))-\abs{f(\veps)}\LandauO(\sqrt{\veps})\Big)+\LandauO(\veps). 
\end{align}
Inserting the bounds $\abs{f(\veps)}\leq \LandauO(\sqrt{\veps})$ and $g(\veps)\leq \LandauO(\veps)$ leads to the lower bound of $1-\LandauO(\veps)$ on the inner product of the Choi states of $\Lambda_\iS$ and the target unitary channel with the unitary operator $U\Sgate U^\dagger$.
This inner product is sometimes referred to as the entanglement fidelity~\cite{Kliesch2020TheoryOfQuantum}, which is related to the average gate fidelity through a known relation~\cite{Kliesch2020TheoryOfQuantum},
\begin{equation}\label{app:eq:proof_S_fidelity}
\Fid(\Lambda_\iS,U\Sgate U^\dagger) = \frac{2}{3}\Tr\left[\Choi(\Lambda_\iS)\Choi(U\Sgate U^\dagger)\right]+\frac{1}{3}.
\end{equation}
\Cref{app:eq:proof_S_fidelity} leads directly to the first claim of the theorem in \cref{app:eq:th_selftest_S}.  

The proof for channel $\Lambda_\iSd$ follows exactly the same steps as for channel $\Lambda_\iS$. 
It is important, however, that the lower bound on $\tr\left[\Choi(\Lambda_\iSd)\Choi(U\Sgate^\dagger U^\dagger)\right]$ is shown to hold for the same gauge unitary $U$ in \cref{app:eq:proof_S_U}.
The main difference from the case of $\Lambda_\iS$, is that roles of states $\phi$ and $\phip$ are swapped, and in $\mathrm{ONB}_1$, the matrix representation of $\Choi(\Lambda_\iSd)$ has the leading terms in the block $C$ rather than the block $A$, if we look at \cref{app:eq:proof_S_char}. 
We can notice that the matrix representation of the Choi state vector $\ket{\Choi(U\Sgate^\dagger U^\dagger)}$ in the same basis is
\begin{equation}\label{app:eq:proof_S_ideal_Choi_Sd}
\left[\ket{\Choi(U\Sgate^\dagger U^\dagger)}\right]_{\mathrm{ONB}_1} = \frac{\e^{-\i\frac{\pi}{4}}}{2}
\begin{bmatrix}
\abs{\bk{\psi}{\phi}}-\abs{\bk{\psip}{\phi}}\\
\abs{\bk{\psi}{\phi}}-\abs{\bk{\psip}{\phi}}\\
-(\abs{\bk{\psi}{\phi}}+\abs{\bk{\psip}{\phi}})\\
\abs{\bk{\psi}{\phi}}+\abs{\bk{\psip}{\phi}}
\end{bmatrix},
\end{equation}
again with the leading terms in the lower half of the vector. 
Apart from that, the reasoning is exactly the same, and the second claim in \cref{app:eq:th_selftest_S} follows.

As for the third and fourth claims in \cref{app:eq:th_selftest_S}, we notice that $U\kb{+}{+}U^\dagger = \psi$, and since we already showed that $\tr[\rho\psi]\geq 1-\frac{3}{2}\veps$ when characterizing $\rho$ in \cref{app:eq:proof_S_4}, we directly conclude that $\tr[\rho U\kb{+}{+}U^\dagger] \geq 1-\frac{3}{2}\veps$. 
Since $M_+ = (1-\lambda_+)\psi+\lambda_-\psip$, we also immediately conclude that $\norm*{M_+-U\kb{+}{+}U^\dagger}_\infty = \max\{\lambda_+,\lambda_-\}\leq \frac{\veps}{2}$.
Note, that at the beginning of the proof we rescaled $\veps$ by the factor of $5$.
This finishes the proof.
\end{proof}

\section{Supporting Lemmata}\label{app:lemmata}
In this section of Appendix, we list the supporting lemmata.

\begin{lemma}\label{app:lemma_Horn}
For a \ac{PSD} matrix $\Gamma = \begin{bmatrix} A & B \\ B^\dagger & C \end{bmatrix}$, with $A,C,B\in\CC^{n\times n}$ it holds that
\begin{equation}
\abs{\bra{v}B\ket{w}}^2\leq \bra{v}A\ket{v}\bra{w}C\ket{w}, \quad \forall \ket{v},\ket{w}\in \CC^n. 
\end{equation}
\end{lemma}
\begin{proof}
Let $K \coloneqq \ket{0}\ket{v}\bra{0}+\ket{1}\ket{w}\bra{1} \in \CC^{2n\times 2}$.
We can write $\Gamma = \kb{0}{0}\otimes A + \kb{0}{1}\otimes B+\kb{1}{0}\otimes B^\dagger + \kb{1}{1}\otimes C$, and express $K^\dagger\Gamma K$,
\begin{align}
\begin{split}
K^\dagger\Gamma K &= (\ket{0}\bra{0}\bra{v}+\ket{1}\bra{1}\bra{w})(\kb{0}{0}\otimes A + \kb{0}{1}\otimes B+\kb{1}{0}\otimes B^\dagger + \kb{1}{1}\otimes C)(\ket{0}\ket{v}\bra{0}+\ket{1}\ket{w}\bra{1})\\
&= \kb{0}{0}\bra{v}A\ket{v}+\kb{0}{1}\bra{v}B\ket{w}+\kb{1}{0}\bra{w}B^\dagger\ket{v}+\kb{1}{1}\bra{w}C\ket{w}.
\end{split}
\end{align} 
The matrix representation of $K^\dagger\Gamma K$ in the computational basis is therefore $K^\dagger\Gamma K = \begin{bmatrix} \bra{v}A\ket{v} & \bra{v}B\ket{w} \\ \bra{w}B^\dagger\ket{v} & \bra{w}C\ket{w} \end{bmatrix}$.
Since $\Gamma\geq 0$, then also $K^\dagger\Gamma K\geq 0$, since $K^\dagger(\argdot)K$ is \ac{CP}.
The claim of the lemma then follows from non-negativity of the determinant of $K^\dagger\Gamma K$.
\end{proof}
In Ref.~\cite{horn1985matrix}, the above lemma is stated as part of a theorem (Theorem 7.7.7), which holds for $A>0$ and $C>0$.
Therefore, we preset the proof above for completeness, to account for the cases of non-invertible $A$ and $C$.

\begin{lemma}
\label{lemma:double_channel_purity}
Given a qubit channel $\Lambda: \L(\CC^2)\to\L(\CC^2)$, if $\Lambda \circ \Lambda (\kb{\psi}{\psi})=\kb{\phi}{\phi}$, for any two $\ket{\psi},\ket{\phi}\in \CC^2$, then $\Lambda(\kb{\psi}{\psi})$ is a pure state. 
\end{lemma}
\begin{proof}
Assume the opposite, that is $\Lambda (\kb{\psi}{\psi}) = \lambda\kb{\theta}{\theta}+(1-\lambda)\kb{\theta^\perp}{\theta^\perp}$ for some $\ket{\theta}\in\CC^2$, and $\lambda\in(0,1)$. 
Then, from linearity it must hold that $\lambda\Lambda(\kb{\theta}{\theta})+(1-\lambda)\Lambda(\kb{\theta^\perp}{\theta^\perp}) = \kb{\phi}{\phi}$, which is only possible if $\Lambda(\kb{\theta}{\theta})=\Lambda(\kb{\theta^\perp}{\theta^\perp})$, and hence $\Lambda(\openone) = 2\kb{\phi}{\phi}$, which means that $\Lambda$ is a measure-and-prepare channel, and, in particular, $\Lambda(\kb{\psi}{\psi})= \kb{\phi}{\phi}$.
Indeed, if $\Lambda(\kb{\psi}{\psi})\neq \kb{\phi}{\phi}$, then due to $\Lambda(\openone) = 2\kb{\phi}{\phi}$, we would have that $\Lambda(\kb{\psip}{\psip}) = 2\kb{\phi}{\phi}-\Lambda(\kb{\psi}{\psi})$ is not \ac{PSD}.
We reach the contradiction, because we assumed that $\Lambda(\kb{\psi}{\psi})$ is not pure.
\end{proof}

\begin{lemma}\label{app:lemma_unitarity}
    Let $\Lambda: \L(\CC^2)\to\L(\CC^2)$ be a qubit channel which maps an \ac{ONB} $\{\ket{\psi},\ket{\psi^\perp}\}$ to an \ac{ONB} in $\CC^2$. 
    Let further $\Lambda(\kb{\varphi}{\varphi})$ be a pure state for some other state vector $\ket{\varphi}$, such that $0<\abs{\bk{\varphi}{\psi}}<1$. 
    Then the channel $\Lambda$ is unitary.
\end{lemma}
\begin{proof}
    Let $\Set{\ket{\phi},\ket{\phip}}$ be an \ac{ONB} such that $\kb{\phi}{\phi} = \Lambda(\kb{\psi}{\psi})$ and $\kb{\phip}{\phip} =\Lambda(\kb{\psip}{\psip})$. 
    From the \ac{CPTP} condition, we conclude that $\Lambda(\kb{\psi}{\psip}) = z\kb{\phi}{\phip}$, and $\Lambda(\kb{\psip}{\psi}) = z^\ast\kb{\phip}{\phi}$, for some $z\in\CC$, with $\abs{z}\leq 1$.
    
    Let $\ket{\varphi} = \sqrt{a}\ket{\psi}+\sqrt{1-a}\ket{\psip}$ for some $a\in\RR$ (which we can always achieve by fixing the global phases of $\ket{\psi}$ and $\ket{\psip}$), and write
    \begin{equation}\label{app:eq_lemma_unitarity_1}
        \Lambda(\kb{\varphi}{\varphi})=a\kb{\phi}{\phi}+(1-a)\kb{\phi^\perp}{\phi^\perp}+z \sqrt{a(1-a)}\kb{\phi}{\phi^\perp}+z^\ast \sqrt{a(1-a)}\kb{\phi^\perp}{\phi}.
    \end{equation}
    The purity of $\Lambda(\kb{\varphi}{\varphi})$ in \cref{app:eq_lemma_unitarity_1} leads to the condition
    \begin{equation}
        1=\tr[(\Lambda(\kb{\varphi}{\varphi}))^2]=a^2+(1-a)^2+2a(1-a)\abs{z}^2.
    \end{equation}
    From the assumptions on $\ket{\varphi}$, we have $0<a<1$, and hence $\abs{z}=1$. 
    From here it is straightforward to see that
    $\Lambda(\argdot)=U(\argdot) U^\dagger$ for the unitary $U=\kb{\phi}{\psi}+z^\ast\kb{\phi^\perp}{\psi^\perp}$.
\end{proof}
\end{appendix}

%\clearpage
\twocolumngrid
\bibliographystyle{./myapsrev4-2}
\bibliography{ref,mk}

\end{document}